\documentclass[sigconf]{acmart}
\usepackage{microtype}
\usepackage{amsmath,amssymb,amsthm}
\usepackage{setspace}
\usepackage{paralist}
\usepackage{algorithmic}
\usepackage{tikz}
\usepackage{paralist}
\usepackage{comment}

\newtheorem{lemma}{Lemma}
\newtheorem{theorem}{Theorem}

\newtheorem{proposition}{Proposition}

\newcommand{\eps}{\varepsilon}

\def\todo#1\par{\vskip2ex\par\noindent\hskip-2em\rule{1ex}{1ex}\hskip3em\framebox{\parbox{.9\textwidth}{#1}}\bigskip\par}

\tikzstyle{timeline}=[very thin,color=lightgray]
\tikzstyle{msg}=[->]
\tikzstyle{every picture}+=[remember picture]

\newcommand{\FD}[1]{\mathcal{#1}}

\renewcommand{\geq}{\geqslant}

\renewcommand{\geq}{\geqslant}

\renewcommand{\le}{\leqslant}
\renewcommand{\geq}{\geqslant}
\renewcommand{\ge}{\geqslant}

\newcommand{\ASYNC}{\mathcal{M}^{\text{async}}}

\newcommand{\N}{\mathbb{N}}
\newcommand{\M}{\mathcal{M}}
\newcommand{\Mpar}{\mathcal{M}^{\text{GST}}}

\newcommand{\Time}{\mathbb{T}}

\newcommand{\D}{\mathcal{D}}

\setcopyright{rightsretained}
\acmDOI{10.475/123_4}

\acmISBN{123-4567-24-567/08/06}

\acmConference[ICDCN'19]{International Conference on Distributed Computing and Networking}{January 2019}{Indian Institute of Science, Bangalore, India}
\acmYear{2019}
\copyrightyear{2019}

\author{Martin Biely}
\affiliation{%
  \institution{EPFL}
  \city{Lausanne}
  \country{Switzerland}
}
\email{martin@biely.eu}

\author{Peter Robinson}
\authornote{
Peter Robinson acknowledges the support of the Natural Sciences and Engineering Research Council of Canada (NSERC).
}
\affiliation{%
  \institution{McMaster University}
  \city{Hamilton}
  \country{Canada}
}
\email{peter.robinson@mcmaster.ca}

\title{On the Hardness of the Strongly Dependent Decision Problem}
\begin{document}

\copyrightyear{2019} 
\acmYear{2019} 
\setcopyright{acmcopyright}
\acmConference[ICDCN '19]{International Conference on Distributed Computing and Networking}{January 4--7, 2019}{Bangalore, India}
\acmBooktitle{International Conference on Distributed Computing and Networking (ICDCN '19), January 4--7, 2019, Bangalore, India}
\acmPrice{15.00}
\acmDOI{10.1145/3288599.3288614}
\acmISBN{978-1-4503-6094-4/19/01}

\begin{abstract}
We present necessary and sufficient conditions for solving the strongly dependent decision (SDD) problem in various distributed systems.
Our main contribution is a novel characterization of the SDD problem based on point-set topology.
For partially synchronous systems, we show that any algorithm that solves the
SDD problem induces a set of executions that is closed with respect to the point-set topology.
We also show that the SDD problem is not solvable in the asynchronous system
augmented with any arbitrarily strong failure detectors.
\end{abstract}
\maketitle

\section{Introduction}  \label{sec:intro} \label{sec:sdd}
The Strongly Dependent Decision Problem (SDD) problem was introduced
     in \cite{CGS00}. 
Like many classic distributed  agreement problems (cf.\ \cite{Lyn96}),
     it belongs to the class of decision tasks. 
In this work, we study the hardness of the problem from the viewpoint of point-set topology and also shed some light on the SDD problem with respect to the power of failure detectors. There are several previous works that have applied algebraic topology and point-set topology to problems in distributed computing, e.g., see \cite{topbook,topbook2,nowak,abc} and the references therein.  

\paragraph{Problem Definition.}
We consider two processes $s$ and $d$. 
Process $s$ (the source) starts with an input value taken from the set
     $\{0,1\}$. 
The problem is for $d$ (the destination) to eventually output a
     decision value from the set $\{0,1\}$, such that the following
     three conditions hold:
\begin{itemize}
\item {\em Integrity}\/:\ Process $d$ decides at most once.
\item {\em Validity}\/:\ If $s$ has not initially crashed, $d$ decides
  $s$'s initial value.
\item {\em Termination}\/:\ If $d$ is correct, then $d$ eventually
     decides.
\end{itemize}

\subsection{System Model} \label{sec:model}

We now formally define our basic system model.
We consider a system of $2$ processes $s$ and $d$ that communicate
via message passing, using messages taken from some (possibly infinite) 
universe.
Every process executes an instance of a \emph{distributed algorithm} that is
modeled as a deterministic state machine, which has a local state that initially
contains the input value of the SDD problem.
A \emph{step} of a process $p$ is a state transition of the state machine that
is guided by a transition relation, which atomically takes the current local
state of $p$, a (possibly empty) subset of messages from $p$'s current message
buffer, and, in case of failure detectors, a value from the
failure detector's  domain, and yields a new local state.
Moreover, a deterministic message sending function computes a possibly empty set
of messages that are to be sent to the other processes, which are placed in the
respective message buffers.
A step can either be a \emph{send step} if a message is sent, a \emph{receive step} if some messages are received, a \emph{send-receive step} if both happens, or a
\emph{local step} if no message is sent or received.
In the absence of failure detectors, we say that a step
$\sigma$ of process $p$ is \emph{trivial}, if $p$'s local state
(comprising memory and message buffers) is unchanged due to $\sigma$;
otherwise we call $\sigma$ \emph{non-trivial}.

A \emph{configuration} of the system consists of the vector of local
states and the message buffers of all the processes; in the initial
configuration, all processes are in an initial state and the
message buffers are empty.
An \emph{execution} 
\[
\rho=(C_0,C_1,\dots)
\]
 is an infinite sequence of 
configurations that starts from an initial configuration $C_0$, and, for $i \ge 
0$, $C_{i+1}$ results
from a step of a single process in configuration $C_i$.
Note that if the $i$-th step ($i\ge 1$) is trivial, then $C_{i-1}=C_i$.

The above basic model is strengthened by restricting the set
of executions by some \emph{admissibility} conditions that depend on the
particular system model used. For example, the classic asynchronous model 
(cf.\ \cite{FLP85}), denoted as $\ASYNC$, requires that
every correct process takes an infinite number of steps,
faulty processes execute only finitely many steps,
and every message sent by a process to a correct receiver process 
  is eventually received.
Similarly to the asynchronous model of \cite{CT96}, we assume that
     processes take steps according to some discrete timebase
     $\mathbb{T}$, which corresponds to the non-negative integers.
Consider an execution $\alpha=(C_0,\dots)$.
We say that $k$ is the \emph{decision time} of $d$, if process $d$ has 
decided in $C_k$ and has not yet decided in $C_{k-1}$, and we call $C_k$ a 
\emph{deciding configuration}.

\section{A Necessary Condition in Partially Synchronous Systems}

In this section we consider variants of the partially synchronous system 
model (cf.\ \cite{DLS88,DDS87}), which strengthen the classic asynchronous 
model (cf.\ Section~\ref{sec:model}) by assuming additional guarantees on
process step times and message delivery.
In the spirit of \cite{DLS88,DDS87} and in contrast to 
Section~\ref{sec:fd}, here we assume that processes do \emph{not} have 
access to failure detectors.

A partially synchronous {model} $\M$ corresponds to a (sub)set of
executions in $\ASYNC$, which are exactly the executions that satisfy the
modeling assumptions of $\M$.
By a slight abuse of notation, we use $\M$ to refer to the admissible executions
and the model itself.
In our analysis, we utilize the framework based on point-set topology that was first introduced in \cite{AS84}.

Let $\alpha$ and $\beta$ be executions (i.e. sequences of configurations, cf.
Section~\ref{sec:model}) in $\ASYNC$.
We define a function 
$d:\M\times\M\rightarrow\mathbb{R}$ as
  
\[
d(\alpha,\beta)~:=~2^{-N}
\]
 where $N$ is the first index 
  where the configurations of $\alpha$ and $\beta$ differ, and $d(\alpha,\beta)~:=~0$ if $\alpha=\beta$.
\begin{lemma} 
  Function $d$ is a metric on $\ASYNC$.
\end{lemma}
\begin{proof}
By definition, $d$ is nonnegative and $\forall \alpha,\beta \in \ASYNC$ we have
$d(\alpha,\beta)=d(\beta,\alpha)$.
For $\alpha,\beta,\gamma \in \ASYNC$, the triangle-inequality 
\[
d(\alpha,\beta)
\le d(\alpha,\gamma) + d(\gamma,\beta)
\] 
trivially holds if $\gamma=\alpha$ or
$\gamma=\beta$.
Now consider the case that 
\[
d(\alpha,\gamma) \ge d(\gamma,\beta) > 0. 
\]
This means that, for some indices $n_1\le n_2$, it holds that
\begin{align*}
d(\alpha,\gamma)&=2^{-n_1}, \\
d(\gamma,\beta)&=2^{-n_2}.
\end{align*}
Since $\gamma$ shares a common prefix of length $n_2-1$ with $\beta$ but only a
prefix of length $n_1-1$ with $\alpha$, it follows that $\alpha$ and
$\beta$ differ at index $n_1$, and thus 
\[
d(\alpha,\beta) = d(\alpha,\gamma)
\]
 and
the triangle-inequality follows.
The case where $0<d(\alpha,\gamma) < d(\gamma,\beta)$ follows analogously.
\end{proof}

It is well known that a metric induces a topology (e.g., \cite[page~119]{Munkres}) where the $\eps$-balls  defined as \[
B_\epsilon(\alpha) = \{
\beta \in \ASYNC \mid d(\alpha,\beta) < \epsilon \}
\]
 are the basic open sets.
We first recall some basic definitions from point-set topology that we use below; we refer the reader to \cite{Munkres} for details.
A set is defined to be \emph{closed} if and only if its complement is open.
Moreover, a subset $X$ of the topological space $\ASYNC$ is called \emph{dense} (in $\ASYNC$) if every execution $\alpha \in \ASYNC$ either belongs to $X$ or is a limit point of $X$; in other words, for any $\epsilon>0$, $B_\epsilon(\alpha) \cap X \ne \emptyset$.

\begin{proposition}[e.g., \cite{Munkres}]\label{prop:union}
The union of any (possibly infinite) collection of open sets is open.
The finite intersection of a collection of closed sets is closed. 
\end{proposition}

We will now argue why safety properties correspond to
closed sets and liveness properties correspond to dense
sets.          We emphasize that the following correspondences were also mentioned by Alpern and
  Schneider~\cite{AS84}. However, in contrast to \cite{AS84}, we consider
  these properties in the metric space induced by $d$.
  
\begin{lemma}  
Consider the metric topology on the set of executions $\ASYNC$.  
A safety property defines a closed set, whereas a liveness property corresponds to a dense set. 
\end{lemma}
\begin{proof}
We first show the result for safety properties.
If an execution $\alpha$ does not satisfy a safety property
  $S\subseteq\ASYNC$, i.e.\ $\alpha \notin S$, then there is an index 
  $N$ where
  all executions $\beta$ that share a prefix longer than $N$ with $\alpha$ are
  not in $S$. (This closely matches intuition, since once a 
    safety property is violated in a prefix of an execution, it makes no 
    difference how this prefix is extended.)
Formally speaking, suppose that $\alpha \notin S$. 
There exists an $N\geq 0$ such that, if some
$\beta \in \M$ has 
\[
d(\alpha,\beta) < {2^{-N}}\text{,}
\]
 i.e., $\alpha$ and 
$\beta$ share a prefix of length $\ge N$, then $\beta \notin S$.  
It follows that, for each $\alpha \notin S$, there is an $\epsilon >0$ such that the $\epsilon$-ball $B_\epsilon(\alpha)$ does not intersect with $S$.
The union of the $\epsilon$-balls of all $\alpha \notin S$ precisely contains all executions in $\ASYNC \setminus S$ and, by Proposition~\ref{prop:union}, is an open set. 
Thus, the set of executions $S$ is a {closed} set since its complement is open.  

We next consider liveness properties.
If $L$ is a liveness property then, for any execution $\alpha \in \ASYNC$ and
any finite prefix $\rho$ of $\alpha$, it is possible to extend $\rho$ yielding
an execution $\beta \in L$.
In other words, any given prefix is ``live''.
To show that a liveness property $L$ is a dense set in our metric topology, we
need to show that, for any $\eps>0$ and any $\alpha \in \ASYNC$, the basic open
set $B_\eps(\alpha)$ intersects $L$, i.e., there is an execution $\beta \in L$
such that $d(\alpha,\beta)< \eps$.
For a fixed $\alpha$ and $\eps>0$, let $n$ be the smallest integer such that
$2^{-n}\le \eps$.
Since $L$ is a liveness property, there exists a $\beta \in L$ that shares a
prefix of length $\ge n+1$ with $\alpha$, which shows that
\[
d(\alpha,\beta)<2^{-n}\le \eps
\]
 as required.
\end{proof}

We now consider some of the classic \emph{partially synchronous} models in this
context:
First, note that the synchronous model is entirely determined by safety properties and hence the executions of any algorithm in this model form a closed set.
Note that in the {partial synchrony} classification of \cite{DDS87},
the synchronous model corresponds to parameters $c=1$ (synchronous
communication) and $p=1$ (synchronous processes).
Now, consider the partially synchronous model $\Mpar$ of \cite{DLS88} where every
execution has a global stabilization time $GST$, i.e., before time $GST$
the system can be completely asynchronous but from time $GST$ on,
communication and computation become synchronous.
The executions of the consensus algorithm $A$ of \cite{DLS88} are \emph{not}
closed because the adversary determines $GST$.
In more detail, it is possible to construct a converging sequence
$(\alpha_i)_{i\ge 0}$ of executions of $A$ in this model,
such that $GST$ is strictly increasing over this sequence.
The limit of this sequence $\alpha=\lim_{i\rightarrow\infty}\alpha_i$ is
the case where $GST=\infty$.
Since $\alpha$ violates the assumption of having a finite $GST$, execution
$\alpha$ is not in the set of executions of $A$ in $\Mpar$ (but rather in
$\ASYNC\setminus\Mpar$).
In other words, the set of executions of $A$ in $\Mpar$ does not contain all limit 
points and hence is \emph{not} closed.

In order to solve the SDD problem, an algorithm needs to satisfy
{Integrity}, {Validity}, and {Termination} (cf.\
Sec.~\ref{sec:sdd}).
These properties correspond to sets of executions in
$\ASYNC$; we denote these sets by $I$, $V$, and $T$ respectively.
Clearly $I$, $V$ are closed (w.r.t.\ to the metric space on $\ASYNC$), 
whereas $T$ is a liveness property.
We consider SDD-algorithms that obey the following condition:
\begin{itemize}
\item[(C1)] Process $d$ decides at the latest upon receiving a message
  from $s$ and takes no non-trivial steps (cf.\ Section~\ref{sec:model})
  afterwards.
Moreover, process $s$ takes no non-trivial steps after sending
  a message to $d$ and sends a message to $d$ in its first step.
\end{itemize}
Any algorithm that solves the SDD problem in
the partially synchronous framework of \cite{DLS88,DDS87} can be transformed into an
algorithm satisfying (C1), by initially sending a message $m$ from $s$ to $d$,
omitting all other non-trivial steps at $s$, and omitting all non-trivial steps
at $d$ that occur after the reception of $m$ by $d$.

In terms of the topological framework, we say that an algorithm $A$ \emph{solves
the SDD problem in a model}, if the set of executions $\M$ of $A$ in this model
satisfies 
\[
\mathcal{M}\subseteq I\cap V\cap T,
\]
 i.e., every execution of $A$ in
the model satisfies the three properties of the SDD problem.

\begin{lemma} \label{lem:bound}
  Let $A$ be an algorithm that adheres to (C1) and solves the SDD
  problem in some model and let $\M \subseteq \ASYNC$ be the
  (corresponding) set of executions of $A$.
  Suppose that $\mathcal{M}$ is not closed.
If process $s$ is initially alive,
  then there is no upper bound on the decision time of process $d$,
  independently of whether $s$ starts with $0$ or $1$.
\end{lemma}

\begin{proof}
Assume that $\mathcal{M}$ is not closed. 
Then, $\mathcal{M}$ does not contain at least one of its limit points, i.e., there is a converging sequence of
executions $(\alpha_k)_{k \in \N}$ such that 
\begin{align*}
\forall~k~\in~\N\colon
\alpha_k~\in~\mathcal{M}
\intertext{and}
\lim_{k\rightarrow\infty}\alpha_k\notin\mathcal{M}.
\end{align*}
Note that if a sequence $(\beta_k)_{k\in \N}$ is converging, this implies that
there exists an index $r'$ such that $s$ has the same input value in all
executions $\beta_n$ $(n\ge r')$. 
W.l.o.g., we can assume that $r'=0$ for sequence $(\alpha_k)_{k \in \N}$.
Suppose that process $s$ starts with input value $v$ in execution 
$\alpha_0$ (and also in all other executions in the sequence).
For the sake of a contradiction, assume that there exists an upper bound 
until reaching a deciding configuration and let $N$ be the least upper bound.
By assumption, $s$ is initially alive in every $\alpha_k$,
and thus, according to (C1), process $s$ sends $m$ to $d$ in its first step.
For each $\alpha_k$, we define $N_k$ to be the number of steps taken by $s$ and $d$ until $d$ receives $m$.
Clearly, $N_k$ is exactly the number of steps taken before a deciding
configuration is reached in $\alpha_k$ since, by (C1), $d$ must decide (at the latest) upon receiving $m$.
By assumption, $N$ is an upper bound on $N_k$ for all $k\ge 0$.
Consider execution 
\[ 
\alpha_\ell = (C_0,\dots,C_N,C_{N+1},\dots),
\]
where $N_\ell =
N$.
In other words, $C_N$ is a deciding configuration.
Note that (C1) implies that neither $s$ nor $d$ take any 
\emph{non-trivial} steps (cf.\ Section~\ref{sec:model}) after $d$ has
decided.
Thus, for all $j>N$, it holds that $C_j = C_N$, which means that
\[
\forall j > N\colon \alpha_j = \alpha_\ell.
\]
Hence, 
\[
\lim_{k\rightarrow\infty}\alpha_k = \alpha_\ell,
\]
and since $\alpha_\ell \in \mathcal{M}$, this
 yields a contradiction for the case where $s$ starts with input 
value $v$.

Now consider the case where $s$ starts with $v'=1-v$.
We need to argue that there exists a sequence $(\alpha_k')_{k \in \N}$ in 
$\M$ such that there is no upper bound on the decision time.
We now show how to construct $\alpha_k'$, given $\alpha_k$.
To this end, we will show by induction that we can define the step 
schedule of $\alpha_k'$ to be \emph{similar} as in $\alpha_k$, for any $k \in 
\N$, in the sense that a send (resp.\ receive) step occurs at time $t$ at
process $p$ in $\alpha_k$, for any choice of $p \in \{s,d\}$, if and only if a send (resp.\ receive)
step occurs at time $t$ at process $p$ in $\alpha_k'$.
This will imply that there is no upper bound on the decision time in the sequence $(\alpha_k')_{k \in \N}$.
Note that the actual configurations of executions $\alpha_k$ and $\alpha_k'$, 
however, are not necessarily the same.

By (C1), process $s$ sends $m$ in its first step, regardlessly of having input
value $v$ or $1-v$.
Let $T_k$ be the time when this happens in $\alpha_k$.
Observe that process $d$ has the same view in every execution until it 
receives a message from $s$.
Thus we can schedule the same type of step (send, receive, or local step) to
happen initially in $\alpha_k'$ as in $\alpha_k$.
This shows the induction base.

For the induction step, assume that we have defined similar schedules up to 
time $\tau$.
If $T_k\ge \tau$, i.e., $s$ has not taken any steps yet, we can argue the same 
way as in the induction base.
Now assume that $T_k<\tau$, i.e., process $s$ has sent a message $m'$ to $d$ in some 
previous step.
By (C1), process $s$ only takes trivial steps after $T_k$ and, in particular, 
does not send any other messages to $d$ later on.
Thus we can schedule either $s$ or $d$ to take a step in $\alpha_k'$, 
accordingly to $\alpha_k$, as required.
Moreover, we schedule process $d$ to receive message $m'$ in this step of 
$\alpha_k'$ if and only if $d$ receives $m$ in $\alpha_k$.

Since we have shown that there exists no upper bound on the decision time in 
$(\alpha_k)_{k \in \N}$ when $s$ starts with value $v$, it follows that the 
same is true for the sequence $(\alpha_k')_{k \in \N}$ when $s$ starts with 
value $1-v$.
\end{proof}

\begin{theorem} \label{thm:sdd}
    Let $A$ be an algorithm that solves the SDD problem in a model and
    let $\M\subseteq \ASYNC$ be the corresponding set of executions of $A$
    in this model.
    Then $\M$ is closed.
\end{theorem}
\begin{proof}
    If $A$ solves the SDD problem in model $\M$, then 
    \[
    \mathcal{M} \subseteq I \cap V \cap T
    \]
    That is, it must be that
    \[
    \mathcal{M} = I \cap V \cap T \cap 
    \mathcal{M},
    \]
    since, otherwise, $\mathcal{M}$ would contain an execution 
    \[
    \gamma \in 
    \ASYNC \setminus \left(I \cap V \cap T \right),
    \]
     contradicting the correctness of $A$.
Recalling from Property~\ref{prop:union} that (finite) intersections of closed sets are closed and since $I$ and $V$ are both safety properties,
it follows that $I\cap V$ is closed too.
    Thus, we are
    done if we can show that $T\cap \mathcal{M} = \mathcal{M}$ is closed
    too. 

    Now assume in contradiction that $\mathcal{M}$ is not closed.  
    Consider an execution $\alpha_0$ where $s$ has an input value of $0$ 
    and
    crashes initially.
    Since $A$ solves the SDD problem, $d$ eventually 
    decides in $\alpha_0$ after some time $k$. By Lemma~\ref{lem:bound} there is an
    execution ${\alpha_0}'$ where $s$ has an input value of $0$, is initially
    alive and, since there is no upper bound on the decision time of $d$,
    we can assume that process $d$ decides at some time $k' > k$.
    By (C1), it follows that $d$ has not received any message in 
    $\alpha_0'$ from $s$ before $k'$ and thus process $d$ has the same
    view in ${\alpha_0}'$ as in $\alpha_0$ up to time $k$; by validity, 
    $d$ must decide on $0$ in both executions.

    Now consider the execution $\alpha_1$ where $s$ has an input value of $1$
    but initially crashes and $d$ decides at some time $k$. Again, by using
    Lemma~\ref{lem:bound}, there is an execution ${\alpha_1}'$, where $s$ is
    initially alive and $d$ receives $m$ at some time $k' >k$. By the same
    reasoning as above, $d$ must decide on $1$ in $\alpha_1$ and ${\alpha_1}'$.
    For process $d$, execution $\alpha_1$ is indistinguishible from $\alpha_0$
    up to time $k$,
    so $d$ decides on the same value in $\alpha_1$ and $\alpha_0$,
    which is a contradiction.  %
\end{proof}

\section{Solvability with Failure Detectors} \label{sec:fd}

Failure detectors have been studied extensively in the quest to understand the impact of asynchrony and faults on the solvability power of distributed systems, e.g., see \cite{CHT96, CT96}.
 
In \cite{CGS00}, it was shown that the SDD problem cannot be solved
     in the asynchronous system $\ASYNC$ equipped with the perfect failure 
     detector $\FD{P}$
     (cf.\ \cite{CT96}). 
This result stands in stark contrast to the fact that, in a synchronous system, which is strong enough to implement $\FD{P}$, the SDD problem \emph{can} be solved! 
So far, the question whether there is any failure detector that is strong enough to solve the SDD problem, and if yes, what is the weakest one to do so, remained open. In this section we will close this gap in literature.

\paragraph{{Failure Detectors}}In the context of failure detectors one important notion is that of a
failure pattern, which we now introduce.  For $t \in \mathbb{T}$, the 
\emph{failure pattern} $F(t)$ denotes the
     set of processes that have crashed up to and including time $t$. 
It is important to remember that if $p$ is in $F(t)$ but was not in
     $F(t')$ ($t'<t$) then this does not mean that $p$ takes a step
     between $t'$ and $t$.
Turning to the specific problem at hand, we recall that for the
     Validity property, it is important whether the source $s$ crashes
     initially or not. 
One way to understand ``initial crash of $s$'' in the context of
     failure detectors is that there is no point in time where $s$ is
     not faulty, i.e.,
\begin{equation}\forall t\in\Time\colon s\in F(t).\label{eq:ic:fd}\end{equation}

Another interpretation is that if $s$ crashes initially, then it
     takes no steps. Given the above timebase we can define
     $\Time_p\subseteq\Time$ to be those points in time where $p$
     takes a step. Then the second interpretation becomes: 
\begin{equation}\Time_s=\emptyset.\label{eq:ic:tb}\end{equation}

While Definition (\ref{eq:ic:fd}) is based purely on the failure pattern and is
     therefore well suited for FDs, definition (\ref{eq:ic:tb}) captures the
     intuitive notion that when a process crashes before doing a
     single step (and is therefore unable to leave its initial state)
     it should be considered initially crashed. 
In the following we will show that for both definitions above there is
     no algorithm that solves the SDD problem in the asynchronous
     model augmented with a failure detector. 
In order to do so, we assume there is an algorithm $A$ solving the SDD
     problem in the asynchronous model augmented with some FD
     $\mathcal{D}$.

We consider (\ref{eq:ic:fd}) first and assume executions $\alpha_v$
     $v\in\{0,1\}$, with a unique $t_c>0$ such that
     process $s$ crashes at $t_c$ and does not take a step before
     $t_c$. 
The two executions are assumed to be identical (step times,
     failure pattern, and FD history), except that in $\alpha_v$ $s$
     has initial value $v$. 
Due to Termination, $d$ has to decide on some
     value $v\in\{0,1\}$ at some time. 
Since the failure detector history can---by definition---only depend
     on the failure pattern, and $d$ queries $\mathcal{D}$ at the same
     times in both, it follows that process $d$ cannot distinguish the
     two execution and thus decides at the same time $t_d$ and the
     same value $w$ in both executions. 
Now assume another execution $\alpha'$ in which $s$ actually performs
     a step before $t_c$, any message it sends is delayed until after
     $t_d$ and that is otherwise (step times, failure pattern, FD
     history and initial value) the same as $\alpha_{1-w}$. 
Clearly in $\alpha'$ $s$ does \emph{not} crash initially, so $d$ has to decide on $1-w$. 
But since up to $t_c$ the execution is indistinguishable from
     $\alpha'$ and $\alpha$, $d$ once again decides on $w$,
     thereby violating Validity.

Now consider case (\ref{eq:ic:tb}):
We start by considering executions $\beta_v$ (for $v\in\{0,1\}$) in
     which $s$ has initial value $v$, process $s$ does not take any
     steps, and that have a common failure pattern such that  
     \[
     \exists
     t_c>0:\ \forall t<t_c:\ s\not\in F(t). 
     \]
Moreover, assume that the step times of $d$ are equivalent in
     $\beta_0$  and $\beta_1$. 
Clearly, in both executions the system's behavior is such that $s$
     does not crash initially, according to case (\ref{eq:ic:tb}). 
Since both executions share the same failure pattern, we can assume
     they also share the same failure detector history, thus process
     $d$ cannot distinguish between the two executions. 
Since due to Termination, $d$ has to decide eventually, it must decide
     the same way in both executions, thus violating Validity and
     leading to a contradiction to Validity for case (\ref{eq:ic:tb}) as well. 
Note that our argument holds for \emph{any} failure detector, i.e., we
     have shown the following result:
\begin{theorem}\label{thm:SDD-FD-impos}
There is no algorithm that solves the Strongly Dependent Decision Problem 
in the asynchronous model augmented with any failure detector $\D$.
\end{theorem}

\section{Conclusion}
We have analyzed the strongly dependent decision problem from a topological angle, which allowed us to succinctly capture the necessary properties of message passing models where the problem is solvable. We believe that a similar approach can be useful for characterizing the properties of system models for other problems in the context of fault-tolerant distributed system.

\providecommand{\bysame}{\leavevmode\hbox to3em{\hrulefill}\thinspace}
\providecommand{\MR}{\relax\ifhmode\unskip\space\fi MR }
\providecommand{\MRhref}[2]{%
  \href{http://www.ams.org/mathscinet-getitem?mr=#1}{#2}
}
\providecommand{\href}[2]{#2}

\end{document}